\documentclass[english]{cccconf}
\usepackage[comma,numbers,square,sort&compress]{natbib}
\usepackage{epstopdf}

\usepackage{amsmath}
\usepackage{amsfonts}
\usepackage{amssymb}
\usepackage{mathrsfs}
\usepackage{color}
\usepackage{hyperref}
\usepackage{booktabs}
\usepackage{graphicx}
\usepackage{epstopdf}
\usepackage{epsfig}
\usepackage{enumerate}
\usepackage{amsopn}
\usepackage{ntheorem}
\usepackage{bm}
\usepackage{algorithm}
\usepackage{algpseudocode}

\newtheorem{lemma}{Lemma}
\newtheorem{theorem}{Theorem}
\newtheorem{remark}{Remark}
\newtheorem{definition}{Definition}
\newtheorem{assumption}{Assumption}
\newtheorem{problem}{Problem}
\newenvironment{proof}{{\noindent\it \textbf{Proof}}\quad}{\hfill $\square$\par}

\title{On the Performance Analysis of Binary Hypothesis Testing with Byzantine Sensors}

 \author{Yuqing Ni$^{1}$, Kemi Ding$^{2}$, Yong Yang$^{3}$, Ling Shi$^{1}$}
 \affiliation{1. Department of Electronic and Computer Engineering, Hong Kong University of Science and Technology, Hong Kong\email{yniac@connect.ust.hk, eesling@ust.hk}}
 \affiliation{2. School of Electrical, Computer and Energy Engineering, Arizona State University, United States of America\email{kding11@asu.edu}}
 \affiliation{3. School of Mechatronic Engineering, Guangdong Polytechnic Normal University, China\email{yy2008@gpnu.edu.cn}}

\DeclareMathOperator{\Tr}{Tr}

\begin{document} \maketitle
\begin{abstract}
We investigate the impact of Byzantine attacks in distributed detection under binary hypothesis testing. It is assumed that a fraction of the transmitted sensor measurements are compromised by the injected data from a Byzantine attacker, whose purpose is to confuse the decision maker at the fusion center. From the perspective of a Byzantine attacker, under the injection energy constraint, an optimization problem is formulated to maximize the asymptotic missed detection error probability, which is based on the Kullback-Leibler divergence. The properties of the optimal attack strategy are analyzed by convex optimization and parametric optimization methods. Based on the derived theoretic results, a coordinate descent algorithm is proposed to search the optimal attack solution. Simulation examples are provided to illustrate the effectiveness of the obtained attack strategy.
\end{abstract}


\footnotetext{The work by Y. Ni and L. Shi is supported by a Hong Kong RGC General Research Fund 16208517.}
	
\section{Introduction}
Wireless sensor networks (WSNs) deploy a large number of sensors to monitor their environment and transmit their measurements to a remote fusion center over wireless communication links. They have been extensively applied in health care monitoring, environmental sensing and industrial monitoring. Based on these received measurements, the fusion center makes a decision about the presence or absence of the phenomenon of interest. Distributed detection at the fusion center has been well studied in detection theory literature~\cite{varshney2012distributed,viswanathan1997distributed}. 

However, these sensors are vulnerable to malicious attacks due to their own limited capabilities and the distributed nature of WSNs. One typical attack type is Byzantine attack. According to~\cite{vempaty2013distributed}, Byzantine attack refers to tampering or falsifying the transmitted data by some internal adversary who has the knowledge about the WSNs. The purpose of the Byzantine attackers is to confuse the fusion center and let the fusion center make an incorrect decision about the state of nature. Distributed detection in the presence of Byzantine attacks has been widely studied in state-of-the-art works. Marano \textit{et al}.~\cite{marano2009distributed} considered the distributed detection under the Neyman-Pearson setup, where a fraction of the sensors were compromised by a Byzantine attacker. An optimal attack strategy to minimize the detection error exponent, which is based on the Kullback-Leibler divergence, was obtained by using a ``water-filling'' procedure. Rawat \textit{et al}.~\cite{rawat2011collaborative} analyzed the performance limits of collaborative spectrum sensing with the presence of Byzantine attackers, who did not know the true state of nature. Optimal strategies for the Byzantine attackers and the fusion center were derived under a minimax game framework. Kailkhura \textit{et al}.~\cite{kailkhura2015asymptotic} adopted Chernoff information as the performance metric and obtained closed-form expressions for the optimal attack strategies which degraded the detection performance most in the asymptotic regime. 

All the works discussed so far for distributed detection under Byzantine attacks consider scenarios where the values of transmitted measurements can only be chosen from a discrete finite alphabet, i.e., $\{0,1\}$. We consider a more general case where the measurement can be any real number. Furthermore, a constraint for the attack power is taken into consideration in our work. We are interested in analytically characterizing the impact of the malicious data injected by a Byzantine attacker. Specifically, from the Byzantine attacker's perspective, what is the most effective attack strategy under limited injection power? 

In this work, we adopt a standard model in distributed detection under binary hypotheses $\mathcal{H}_0$ versus $\mathcal{H}_1$ with known Gaussian distributions. Measurements are independently and identically distributed conditioned on the unknown hypothesis. We assume that the Byzantine attacker knows the true state of nature and they inject independent Gaussian noises to a fraction of the measurements based on this knowledge. The fusion center makes the detection under the Neyman-Pearson setup.

The remainder of this paper is organized as follows: Section~\ref{sec:formulation} introduces the Byzantine attack model and the problem of interest. Section~\ref{sec:pre} provides some preliminaries about the approximation methods of the KL divergence between Gaussian mixture models. Section~\ref{sec:main} presents the main theoretic results regarding the optimal attack strategy and proposes an algorithm to search the optimal solution. Section~\ref{sec:numerical} shows simulation examples and gives interpretations. Section~\ref{sec:conclusion} draws conclusions.

\emph{Notations}: $\mathbb{R}$ denotes the set of real numbers. $\mathbb{R}^n$ is the $n$-dimensional Euclidean space. $\mathbb{S}_{+}^{n}$ ($\mathbb{S}_{++}^{n}$) is the set of $n\times n$ positive semi-definite (definite) matrices. When $X\in \mathbb{S}_{+}^{n}$ ($\mathbb{S}_{++}^{n}$) , we simply write $X\succeq0$ ($X\succ0$). $\displaystyle \mathscr{N}\left(\mu,~\Sigma\right)$ denotes a Gaussian distribution with mean $\mu$ and variance $\Sigma$. The notation $\sim$ is read as ``is distributed according to''. $\Tr(\cdot)$ stands for the trace of a matrix. $\Vert\cdot\Vert$ and the superscript $\left(\cdot\right)^\top$ denote the Euclidean norm and the transpose of a vector, respectively.
	
\section{Problem Formulation}\label{sec:formulation}
Consider a binary state detection problem, where $\theta\in\{0,~1\}$, using $m$ sensors' measurements. Define the measurement from sensor $j$ as $x_j\in\mathbb{R}^n$. Given the state $\theta$, we assume that all measurements $\{x_j\}_{j=1,2,\dots,m}$ are independently and identically distributed (i.i.d.). When the state $\theta=0$, the probability measure generated by $x_j$ is $f_0$ and when $\theta=1$, it is denoted as $f_1$. We assume that the probability measures $f_0$ and $f_1$ are Gaussian distributions under two hypotheses $\mathcal{H}_0$ and $\mathcal{H}_1$: \begin{align*}
\mathcal{H}_0:~~&f_0\sim \mathscr{N}\left(\mu_0,~\Sigma_0\right),\\
\mathcal{H}_1:~~&f_1\sim \mathscr{N}\left(\mu_1,~\Sigma_1\right),
\end{align*}
where $\Sigma_0, \Sigma_1\succ0$.

\subsection{Byzantine attack model}
Denote the manipulated measurements at sensor $j$ as
\begin{align*}
x^\star_j=x_j+x_j^a,
\end{align*}
where $x_j^a\in \mathbb{R}^n$ is the bias vector injected by the attacker obeying Gaussian distributions under two hypotheses:
\begin{align*}
\mathcal{H}_0:~~&f_0^a\sim \mathscr{N}\left(\nu_0-\mu_0,~\Gamma_0-\Sigma_0\right),\\
\mathcal{H}_1:~~&f_1^a\sim \mathscr{N}\left(\nu_1-\mu_1,~\Gamma_1-\Sigma_1\right).
\end{align*}  
Assume that the injected bias $x_j^a$ is independent of the original measurement $x_j$. Furthermore, $\Gamma_0\succeq\Sigma_0$ and $\Gamma_1\succeq\Sigma_1$. Correspondingly, the manipulated measurement $x_j^\star$ is also Gaussian distributed. Its probability measures under two hypotheses $\mathcal{H}_0$ and $\mathcal{H}_1$ are given by
\begin{align*}
\mathcal{H}_0:~~&g_0\sim \mathscr{N}\left(\nu_0,~\Gamma_0\right),\\
\mathcal{H}_1:~~&g_1\sim \mathscr{N}\left(\nu_1,~\Gamma_1\right).
\end{align*}
The following assumption is made on the attacker.
\begin{assumption}
(Model Knowledge): The attacker knows the probability measures $f_0$ and $f_1$ and the true state $\theta$.
\end{assumption}
Generally, this is a common assumption regarding the worst-case attacks, which is also included in~\cite{ren2018binary,ren2018secure,marano2009distributed,fellouris2018efficient}. Moreover, this assumption is in accordance with the Shannon's maxim, that is the defensive systems should be designed under the assumption that the enemy will immediately gain full knowledge of the systems. Therefore, the probability measures $f_0$ and $f_1$ can be developed by the attacker. The true state can be obtained by deploying attacker's own sensor network. Based on the model knowledge, the attacker is capable of well designing the injected vectors to confuse the fusion center.
Let the parameter $\alpha\in\left(0,1\right)$ represent the \textit{attacking power} of the adversary. We assume that the $m$ measurements received at the fusion center are manipulated by the attacker with probability $\alpha$. Therefore, the $j$-th sample at the fusion center is distributed as follows:
\begin{align*}
\mathcal{H}_0:~~&\left(1-\alpha\right)f_0+\alpha g_0,\\
\mathcal{H}_1:~~&\left(1-\alpha\right)f_1+\alpha g_1.
\end{align*}
Note that all of these $m$ measurements are conditional i.i.d.. 

\subsection{Problem of interest}
The attacker aims at devastating the detection performance at the fusion center. Similar to~\cite{marano2009distributed} and~\cite{coutino2018submodular}, we quantify the impact of Byzantine attacks by Kullback-Leibler (KL) divergence, which measures the ``distance'' between the hypotheses under test. The KL divergence $\displaystyle \mathcal{D}\left(\left(1-\alpha\right)f_0+\alpha g_0~\Vert~\left(1-\alpha\right)f_1+\alpha g_1\right)$ determines the missed detection error probability under the Neyman-Pearson setup by Stein's lemma~\cite{cover2012elements}. A smaller KL divergence implies a larger missed detection error probability at the fusion center. The attacker should choose $f_0^a$ and $f_1^a$ wisely to minimize the KL divergence under an injection energy constraint. We consider the following optimization problem from the perspective of the Byzantine attacker:
\begin{problem}
	\label{prob:original}
\begin{align*}
\min_{f_0^a,f_1^a,\alpha}~&\mathcal{D}\left(\left(1-\alpha\right)f_0+\alpha g_0~\Vert~\left(1-\alpha\right)f_1+\alpha g_1\right),\\
\rm s.t.~~~~&0<\alpha<1,~~\Gamma_0\succeq\Sigma_0,~~\Gamma_1\succeq\Sigma_1,\\
&\alpha\left[\right.\Tr\left(\Gamma_0+\Gamma_1-\Sigma_0-\Sigma_1\right)\\
&~~~+\Vert\nu_0-\mu_0\Vert^2+\Vert\nu_1-\mu_1\Vert^2\left.\right]\leq \delta,
\end{align*}
\end{problem}
where $\delta$ is a given positive constant, denoting the degree of difficulty for the Byzantine attack. A larger $\delta$ allows more energy to inject, which avails the attacker of more opportunities to launch the Byzantine attack.

\section{Preliminary: KL Divergence Approximation between Gaussian Mixture Models}\label{sec:pre}
In this section, we introduce several methods to approximate the KL divergence between two Gaussian mixtures, which is a key supporting technique to deal with the objective in Problem~\ref{prob:original}, since there is no accurate closed-form expression.

\subsection{Monte Carlo sampling}\label{Monte_Carlo}
For large dimension $n$, Monte Carlo simulation is the only method that can estimate $\displaystyle \mathcal{D}\left(\left(1-\alpha\right)f_0+\alpha g_0~\Vert~\left(1-\alpha\right)f_1+\alpha g_1\right)$ with arbitrary accuracy. We can draw i.i.d. samples $\{z_i\}$ from the probability density function $\left(1-\alpha\right)f_0+\alpha g_0$, and we have~\cite{hershey2007approximating}:
\begin{align*}
\lim_{K\to\infty}\frac{1}{K}\sum_{i=1}^{K}&\log{\frac{\left[\left(1-\alpha\right)f_0+\alpha g_0\right](z_i)}{\left[\left(1-\alpha\right)f_1+\alpha g_1\right](z_i)}}\\
&\to \mathcal{D}\left(\left(1-\alpha\right)f_0+\alpha g_0~\Vert~\left(1-\alpha\right)f_1+\alpha g_1\right).
\end{align*}

\subsection{Upper bound approximation}
By the chain rule for relative entropy~\cite{cover2012elements}, the upper bound of the KL divergence can be given by:
\begin{align*}
\mathcal{D}&\left(\left(1-\alpha\right)f_0+\alpha g_0~\Vert~\left(1-\alpha\right)f_1+\alpha g_1\right)\\
&~~~~~~~~~~~~~~~~~~~~~~\leq\left(1-\alpha\right)\mathcal{D}\left(f_0~\Vert~f_1\right)+\alpha\mathcal{D}\left(g_0~\Vert~g_1\right).
\end{align*}

\subsection{Gaussian approximation}
A common method is to replace the Gaussian mixtures with modified Gaussian distributions~\cite{hershey2007approximating}. Denote the Gaussian approximations as $y_{a_0}$ and $y_{a_1}$:
\begin{small}
	\begin{align*}
	\mathcal{H}_0:~~y_{a_0}\sim &\mathscr{N}\left(\right.\left(1-\alpha\right)\mu_0+\alpha\nu_0,~\left(1-\alpha\right)\Sigma_0+\alpha\Gamma_0\\
	&~~~~+\alpha\left(1-\alpha\right)\left(\mu_0-\nu_0\right)\left(\mu_0-\nu_0\right)^\top\left.\right),\\
	\mathcal{H}_1:~~y_{a_1}\sim &\mathscr{N}\left(\right.\left(1-\alpha\right)\mu_1+\alpha\nu_1,~\left(1-\alpha\right)\Sigma_1+\alpha\Gamma_1\\
	&~~~~+\alpha\left(1-\alpha\right)\left(\mu_1-\nu_1\right)\left(\mu_1-\nu_1\right)^\top\left.\right).
	\end{align*}
\end{small}
Based on this Gaussian approximation method, the KL divergence between two Gaussian mixture models then can be expressed in a closed form~\cite{duchi2007derivations}.

The above three approximations have their own features. The Monte Carlo sampling performs much better in accuracy, especially for high-dimension cases. The upper bound approximation is more concise, but somewhat loose. The Gaussian approximation is a closed-form expression and probably, it tends to be followed by more theoretic analysis. In the following sections, we mainly focus on the Gaussian approximation and derive some theoretic results.    

\section{Main Results}\label{sec:main}
Due to the complexity of Problem~\ref{prob:original}, in this paper, we only consider the scalar case $n=1$, aiming to get some inspiring insights. By the Gaussian approximation, the KL divergence objective is then transformed into:
\begin{small}
\begin{align*}
\mathcal{D}\left(y_{a_0}\| y_{a_1}\right)&=\frac{1}{2}\bigg{[}\frac{\left(1-\alpha\right)\Sigma_0+\alpha\Gamma_0+\alpha\left(1-\alpha\right)\left(\mu_0-\nu_0\right)^2}{\left(1-\alpha\right)\Sigma_1+\alpha\Gamma_1+\alpha\left(1-\alpha\right)\left(\mu_1-\nu_1\right)^2}\\
&+\frac{\left[\left(1-\alpha\right)\mu_1+\alpha\nu_1-\left(1-\alpha\right)\mu_0-\alpha\nu_0\right]^2}{\left(1-\alpha\right)\Sigma_1+\alpha\Gamma_1+\alpha\left(1-\alpha\right)\left(\mu_1-\nu_1\right)^2}-1\\
&-\ln\frac{\left(1-\alpha\right)\Sigma_0+\alpha\Gamma_0+\alpha\left(1-\alpha\right)\left(\mu_0-\nu_0\right)^2}{\left(1-\alpha\right)\Sigma_1+\alpha\Gamma_1+\alpha\left(1-\alpha\right)\left(\mu_1-\nu_1\right)^2}\bigg{]}.
\end{align*}
\end{small}
The problem is complex with all the decision variables $\nu_0$, $\nu_1$, $\Gamma_0$, $\Gamma_1$, and $\alpha$. To deal with this challenging situation, we mildly simplify it by fixing variables $\nu_0$, $\nu_1$, and $\alpha$ first, and we show that it can be transformed into a convex optimization by change of variables with respect to Gaussian variances $\Gamma_0$ and $\Gamma_1$. Second, we reduce the solution space to a search space only depending upon the Gaussian means $\nu_0$ and $\nu_1$, and the \textit{attacking power} $\alpha$. By proving that the new objective is continuous at the above three variables, we reveal the special characteristics of the optimal attack solution. Finally, a coordinate descent algorithm is proposed to search the optimal Byzantine attack policy.
 
\subsection{Results regarding $\Gamma_0$ and $\Gamma_1$}
In this subsection, we fix the Gaussian means $\nu_0$ and $\nu_1$, and the \textit{attacking power} $\alpha$. For notational convenience, we define the following constants:
\begin{small}
\begin{align*}
c_0 &\triangleq \frac{\left(1-\alpha\right)\Sigma_0+\alpha\left(1-\alpha\right)\left(\mu_0-\nu_0\right)^2}{\alpha}>0,\\
c_1 &\triangleq \frac{\left(1-\alpha\right)\Sigma_1+\alpha\left(1-\alpha\right)\left(\mu_1-\nu_1\right)^2}{\alpha}>0,\\
c_2 &\triangleq \frac{\left[\left(1-\alpha\right)\mu_1+\alpha\nu_1-\left(1-\alpha\right)\mu_0-\alpha\nu_0\right]^2}{\alpha}\geq 0.
\end{align*}
\end{small}
The Byzantine attack optimization problem is then transformed into
\begin{problem}
	\label{prob:gaussian_approx_scalar}
	\begin{align*}
	\min_{\Gamma_0,\Gamma_1}~~&\frac{1}{2}\left(\frac{\Gamma_0+c_0}{\Gamma_1+c_1}+\frac{c_2}{\Gamma_1+c_1}-1-\ln\frac{\Gamma_0+c_0}{\Gamma_1+c_1}\right),\\
	\rm s.t.~~~&\alpha\left[\right.\Gamma_0+\Gamma_1-\Sigma_0-\Sigma_1\\
	&~~~+\left(\nu_0-\mu_0\right)^2+\left(\nu_1-\mu_1\right)^2\left.\right]\leq \delta,\\
	&\Gamma_0\geq \Sigma_0,~~\Gamma_1\geq \Sigma_1.
\end{align*}
\end{problem}
To make the problem feasible, we further assume that the given variables satisfy 
\begin{align*}
\delta\geq \alpha\left[\left(\nu_0-\mu_0\right)^2+\left(\nu_1-\mu_1\right)^2\right],~~0<\alpha<1.
\end{align*}
We propose another attack optimization Problem~\ref{prob:gaussian_approx_scalar_change} and give the following Theorem~\ref{thm:equivalent}.
\begin{problem}
	\label{prob:gaussian_approx_scalar_change}
	\begin{align*}
	\min_{\widetilde{\Gamma}_0,\widetilde{\Gamma}_1}~~&\frac{1}{2}\left(\widetilde{\Gamma}_0-\ln\widetilde{\Gamma}_0+c_2\widetilde{\Gamma}_1-1\right),\\
	\rm s.t.~~~&\widetilde{\Gamma}_0\geq\left(\Sigma_0+c_0\right)\widetilde{\Gamma}_1,\\
	&\widetilde{\Gamma}_0\leq\bigg{[}\frac{\delta}{\alpha}-\left(\nu_0-\mu_0\right)^2-\left(\nu_1-\mu_1\right)^2+\Sigma_0+\Sigma_1\\
	&~~~~~~~+c_0+c_1\bigg{]}\widetilde{\Gamma}_1-1,\\
	&\widetilde{\Gamma}_1\leq\frac{1}{\Sigma_1+c_1}.
	\end{align*}
\end{problem}	

\begin{theorem}
	\label{thm:equivalent}
	Problem~\ref{prob:gaussian_approx_scalar} is equivalent to Problem~\ref{prob:gaussian_approx_scalar_change}, which is a convex optimization problem.
\end{theorem}

\begin{proof}
By the change of variables, i.e., $\displaystyle \widetilde{\Gamma}_0\triangleq \frac{\Gamma_0+c_0}{\Gamma_1+c_1}$ and $\displaystyle \widetilde{\Gamma}_1\triangleq \frac{1}{\Gamma_1+c_1}$, it is trivial to verify that Problem~\ref{prob:gaussian_approx_scalar} is equivalent to Problem~\ref{prob:gaussian_approx_scalar_change}. Moreover, the objective of Problem~\ref{prob:gaussian_approx_scalar_change} is convex and constraints are also convex, which shows that it is a convex optimization.
\end{proof}

As a result, we can use the existing algorithms, i.e., gradient descent, to obtain the optimal solutions to Problem~\ref{prob:gaussian_approx_scalar_change}, instead of endeavoring to solve the general Problem~\ref{prob:gaussian_approx_scalar}.

\subsection{Results regarding $\nu_0$, $\nu_1$, and $\alpha$}
For the remaining three variables $\nu_0$, $\nu_1$ and $\alpha$, we will show that there exist some good properties for the optimal solutions. It is treated as a parametric optimization problem. Before that, some preliminaries are presented first. We give the following terms, definitions and Lemma~\ref{lem:berge}, mainly based on~\cite{sundaram1996first} and \cite{aubin2009set}.

\begin{definition}
	Let $S$ and $\Psi$ be subsets of $\mathbb{R}^\ell$ and $\mathbb{R}^m$, respectively. A \textit{correspondence} $\mathcal{C}$ from $\Psi$ to $S$ is a map that associates each element $\psi\in\Psi$ with a nonempty subset $\mathcal{C}\left(\psi\right)\subset S$. We denote such a \textit{correspondence} as $\displaystyle \mathcal{C}: \Psi\rightrightarrows S$.
\end{definition}

\begin{definition}
	A \textit{correspondence} $\displaystyle \mathcal{C}: \Psi\rightrightarrows S$ is upper-semicontinuous at $\psi\in\Psi$ if and only if for any open set $V$ such that $\mathcal{C}(\psi)\subset V$, there exists an open set $U$ containing $\psi$, such that for any $\psi'\in U\cap\Psi$, $\mathcal{C}(\psi')\subset V$ holds. It is said to be upper-semicontinuous on $\Psi$ if and only if it is upper-semicontinuous at each $\psi\in\Psi$.
\end{definition}

\begin{definition}
	A \textit{correspondence} $\displaystyle \mathcal{C}: \Psi\rightrightarrows S$ is lower-semicontinuous at $\psi\in\Psi$ if and only if for any open set $V$ such that $V\cap\mathcal{C}(\psi)\neq \emptyset$, there exists an open set $U$ containing $\psi$ such that for any $\psi'\in U\cap\Psi$, $V\cap \mathcal{C}(\psi')\neq\emptyset$ holds. It is said to be lower-semicontinuous on $\Psi$ if and only if it is lower-semicontinuous at each $\psi\in\Psi$.
\end{definition}

\begin{definition}
	A \textit{correspondence} $\displaystyle \mathcal{C}: \Psi\rightrightarrows S$ is continuous on $\Psi$ if and only if $\mathcal{C}$ is both upper-semicontinuous and lower-semicontinuous on $\Psi$.
\end{definition}

\begin{definition}
	A \textit{correspondence} $\displaystyle \mathcal{C}: \Psi\rightrightarrows S$ is said to be
	\begin{enumerate}
		\item compact-valued at $\psi\in\Psi$ if $\mathcal{C}(\psi)$ is a compact set;
		\item convex-valued at $\psi\in\Psi$ if $\mathcal{C}(\psi)$ is a convex set.
	\end{enumerate}
A \textit{correspondence} $\mathcal{C}$ is said to be compact-valued (convex-valued) if it is compact-valued (convex-valued) at each $\psi\in\Psi$.
\end{definition}

\begin{lemma}
	\label{lem:berge}
	(Berge's Maximum Theorem under Convexity) Let $\displaystyle f: S\times \Psi\to\mathbb{R}$ be a continuous function, and $f(\cdot,  \psi)$ is convex in $s\in S$ for each given $\psi\in\Psi$. Let $\displaystyle \mathcal{C}: \Psi\rightrightarrows S$ be a continuous, compact-valued, and convex-valued \textit{correspondence}. Let $\displaystyle f^\star: \Psi\to\mathbb{R}$ and $\displaystyle \mathcal{C}^\star: \Psi\rightrightarrows S$ be defined as:
	\begin{align}
	f^\star(\psi)&\triangleq\min_{s\in S}\left\{f(s,\psi)\mid s\in\mathcal{C}(\psi)\right\},\label{eq:f*}\\
	\mathcal{C}^\star(\psi)&\triangleq \left\{s\in\mathcal{C}(\psi)\mid f(s,\psi)=f^\star(\psi)\right\}.\label{eq:c*}
	\end{align}
	Then $f^\star$ is a continuous function on $\Psi$, and $\mathcal{C}^\star$ is an upper-semicontinuous, compacted-valued, and convex-valued correspondence on $\Psi$.
\end{lemma}

Lemma~\ref{lem:berge} is a variant of the Berge's maximum theorem. One can find the proof from Theorem $9.17$ in~\cite{sundaram1996first}.

Based on the above preliminaries, we denote two variables as $\displaystyle s\triangleq\left[\widetilde{\Gamma}_0~~\widetilde{\Gamma}_1\right]^\top$ and $\displaystyle \psi\triangleq\left[\nu_0~~\nu_1~~\alpha\right]^\top$. A subset $S$ of $\mathbb{R}^2$ is described as $\displaystyle S=\left\{s~\middle\vert~s>0\right\}$, and a subset $\Psi$ of $\mathbb{R}^3$ is described as $\displaystyle \Psi=\left\{\psi~\middle\vert~\alpha\left[\left(\nu_0-\mu_0\right)^2+\left(\nu_1-\mu_1\right)^2\right]\leq\delta,~0<\alpha<1\right\}$. A continuous function $f:~S\times\Psi\to\mathbb{R}$ is defined as:
\begin{small}
\begin{align*}
f(s,~\psi)\triangleq &~\frac{1}{2}\bigg{[}\widetilde{\Gamma}_0-\ln\widetilde{\Gamma}_0-1\\
&~~~+\frac{\left[\left(1-\alpha\right)\mu_1+\alpha\nu_1-\left(1-\alpha\right)\mu_0-\alpha\nu_0\right]^2}{\alpha}\widetilde{\Gamma}_1\bigg{]}.
\end{align*}
\end{small}
For notational convenience, we define the following two functions $d_0,~d_1:~\Psi\to\mathbb{R}$ as:
\begin{align*}
d_0\left(\psi\right)&\triangleq~\frac{\left(1-\alpha\right)\Sigma_0+\alpha\left(1-\alpha\right)\left(\mu_0-\nu_0\right)^2}{\alpha},\\
d_1\left(\psi\right)&\triangleq~\frac{\left(1-\alpha\right)\Sigma_1+\alpha\left(1-\alpha\right)\left(\mu_1-\nu_1\right)^2}{\alpha}.
\end{align*}
A \textit{correspondence} $\displaystyle \mathcal{C}: \Psi\rightrightarrows S$ is defined as:
\begin{small}
\begin{align*}
\mathcal{C}(\psi)\triangleq~\bigg{\{}s~\mid&~~\widetilde{\Gamma}_0\geq\left(\Sigma_0+d_0\left(\psi\right)\right)\widetilde{\Gamma}_1,\\
&~~\widetilde{\Gamma}_0\leq\bigg{[}\frac{\delta}{\alpha}-\left(\nu_0-\mu_0\right)^2-\left(\nu_1-\mu_1\right)^2\\
&~~~~~~~+\Sigma_0+\Sigma_1+d_0(\psi)+d_1(\psi)\bigg{]}\widetilde{\Gamma}_1-1,\\
&~~\widetilde{\Gamma}_1\leq\frac{1}{\Sigma_1+d_1(\psi)}\bigg{\}}.
\end{align*}
\end{small}
Consider the optimization problem:
\begin{problem}
	\label{prob:correspondence}
	\begin{align*}
	\min_{s,~\psi}~~~~&f\left(s,~\psi\right),\\
	\rm s.t.~~~~&s\in\mathcal{C}(\psi),
	\end{align*}
\end{problem}
where $s\in S$ and $\psi\in\Psi$. The definitions of $f^\star$ and $\mathcal{C}^\star$ are consistent with those in~\eqref{eq:f*} and~\eqref{eq:c*}. Obviously, Problem~\ref{prob:correspondence} is derived from the original Problem~\ref{prob:original} via Gaussian approximation.

\begin{theorem}
	\label{thm:berge_illustration}
	In Problem~\ref{prob:correspondence}, $f^\star$ is a continuous function on $\Psi$, and $\mathcal{C}^\star$ is an upper-semicontinuous, compacted-valued, and convex-valued correspondence on $\Psi$.
\end{theorem}

\begin{proof}
	The proof is mainly based on Lemma~\ref{lem:berge}. It is obvious that $f(\cdot,\psi)$, which is the objective in Problem~\ref{prob:gaussian_approx_scalar_change}, is convex in $s$ for each given $\psi$. For the rest part, we need to check the properties of the \textit{correspondence} $\mathcal{C}$.
	
	Compact-valuedness of $\mathcal{C}$ is obvious, since for each $\psi\in\Psi$, $\mathcal{C}(\psi)$ is closed and bounded. Convex-valuedness is also obvious. In the following, we will show that the \textit{correspondence} $\mathcal{C}$ is both upper-semicontinuous and lower-semicontinuous.
	
	(\emph{Upper-semicontinuous}) Let $V$ be an open set such that $\mathcal{C}(\psi)\subset V$. Define an $\epsilon$-neighborhood $\mathcal{B}_\epsilon(\psi)$ of $\psi$ in $\Psi$ by
	\begin{align*}
	\mathcal{B}_\epsilon(\psi)\triangleq \left\{\psi'\in\Psi \mid \lVert\psi'-\psi\rVert<\epsilon\right\}.
	\end{align*}
	We will prove the upper-semicontinuity by contradiction. Suppose that $\mathcal{C}$ is not upper-semicontinuous at $\psi$. Then $\forall \epsilon>0$, $\exists~s'$ such that $s'\in\mathcal{C}(\psi')$ and $s'\notin V$. Choose a sequence $\epsilon(k)\to 0$, and let $\psi(k)\in\mathcal{B}_{\epsilon(k)}(\psi)$, with $s(k)\in\mathcal{C}(\psi(k))$ but $s(k)\notin V$. We will first show that the $\{s(k)\}$ sequence has a convergent subsequence since the sequence lies in a compact set, which is stated by the Bolzano-Weierstrass theorem~\cite{bartle2000introduction}. Since $\psi(k)\to\psi$, we have $\nu_0(k)\to\nu_0$, $\nu_1(k)\to\nu_1$ and $\alpha(k)\to\alpha$. Therefore, there is $k^\star$ such that for all $k\geq k^\star$, we have
	\begin{align*}
	\lvert\nu_0(k)-\nu_0\rvert\leq\eta,~~\lvert\nu_1(k)-\nu_1\rvert\leq\eta,~~\lvert\alpha(k)-\alpha\rvert\leq\eta,
	\end{align*}
	for some small enough positive $\eta$. By some tedious but basic calculations, it follows that for $k\geq k^\star$, we have $s(k)\in M$, where $M$ is the compact set defined by:
	\begin{small}
	\begin{align*}
	 M&\triangleq\bigg{\{} s\in S \mid \widetilde{\Gamma}_0\geq\Sigma_0\widetilde{\Gamma}_1,~~\widetilde{\Gamma}_1\leq\frac{1}{\Sigma_1},\\
	&\widetilde{\Gamma}_0\leq\bigg{(}\frac{\delta}{\alpha-\eta}+\Sigma_0+\Sigma_1+d_0^{\max}(\psi)+d_1^{\max}(\psi)\bigg{)}\widetilde{\Gamma}_1-1\bigg{\}}.
	\end{align*}
	\end{small}
    For brevity, $d_0^{\max}(\psi)$ and $d_1^{\max}(\psi)$ are denoted as:
    \begin{small}
    \begin{align*}
    d_0^{\max}(\psi)\triangleq&\left(\frac{1}{\alpha-\eta}-1\right)\Sigma_0+\left[1-\left(\alpha-\eta\right)\right]\\
    &~~~\times\left[\left(\mu_0-\nu_0\right)^2+\eta^2+2\eta\lvert\mu_0-\nu_0\rvert\right],\\
    d_1^{\max}(\psi)\triangleq&\left(\frac{1}{\alpha-\eta}-1\right)\Sigma_1+\left[1-\left(\alpha-\eta\right)\right]\\
    &~~~\times\left[\left(\mu_1-\nu_1\right)^2+\eta^2+2\eta\lvert\mu_1-\nu_1\rvert\right].
    \end{align*}
    \end{small}
	Therefore, there is a subsequence of $\{s(k)\}$, which we will continue to denote by $\{s(k)\}$ for notation convenience, converging to a limit $\bar{s}$. Moreover, since $s(k)\in\mathcal{C}(\psi(k))$ and $\psi(k)\to\psi$, $s(k)\to\bar{s}$, we also have $\bar{s}\in\mathcal{C}(\psi)$. Because $\mathcal{C}(\psi)\subset V$, $\bar{s}\in V$ is directly obtained. However, $s(k)\notin V$ for any $k$, and $V$ is an open set. Therefore, we also have $\bar{s}\notin V$, which is a contradiction. This validates the upper-semicontinuity of the \textit{correspondence} $\mathcal{C}$.
	
	(\emph{Lower-semicontinuous}) Let $V$ be an open set such that $V\cap\mathcal{C}(\psi)\neq\emptyset$. Let $s$ be a point in this intersection, and therefore $s\in\mathcal{C}(\psi)$. We denote an internal point of the triangle area characterized by $\mathcal{C}(\psi)$ as $\hat{s}$, i.e.,
	\begin{small}
	\begin{align*}
	\hat{s}\triangleq \begin{bmatrix}
	\frac{\Sigma_0+d_0\left(\psi\right)}{\Sigma_1+d_1\left(\psi\right)}\\
	\frac{\frac{\delta}{\alpha}-\left(\nu_0-\mu_0\right)^2-\left(\nu_1-\mu_1\right)^2+2\left(\Sigma_0+\Sigma_1+d_0(\psi)+d_1(\psi)\right)}{2\left(\Sigma_1+d_1(\psi)\right)\left[\frac{\delta}{\alpha}-\left(\nu_0-\mu_0\right)^2-\left(\nu_1-\mu_1\right)^2+\Sigma_0+\Sigma_1+d_0(\psi)+d_1(\psi)\right]}
	\end{bmatrix}.
	\end{align*}
	\end{small}
    Since $V$ is open, $\displaystyle \kappa s+(1-\kappa)\hat{s}\in V$ for $\kappa<1$, $\kappa$ close to $1$. Let $\tilde{s}\triangleq \kappa s+(1-\kappa)\hat{s}$, and then $\tilde{s}\in\mathcal{C}(\psi)$. We will show the lower-semicontinuity by contradiction. Suppose that $\mathcal{C}(\psi')\cap V=\emptyset$ holds for all $\psi'$ in any neighborhood of $\psi$. Take a sequence $\epsilon(k)\to 0$, and pick $\psi(k)\in\mathcal{B}_{\epsilon(k)}(\psi)$ such that $\mathcal{C}(
    \psi(k))\cap V=\emptyset$. Since $\mathcal{C}(\psi(k))\to\mathcal{C}(\psi)$, for $k$ sufficient large, $\tilde{s}\in\mathcal{C}(\psi(k))$. It implies $\tilde{s}\notin V$, which is a contradiction. 
    
    After proving that the \textit{correspondence} $\mathcal{C}$ is continuous, compact-valued and convex-valued, we conclude that $f^\star$ is continuous and $\mathcal{C}^\star$ is upper-semicontinuous, compacted-valued and convex-valued according to Lemma~\ref{lem:berge}.
\end{proof}

\begin{figure}[t]
	\centering
	\includegraphics[width=0.48\textwidth]{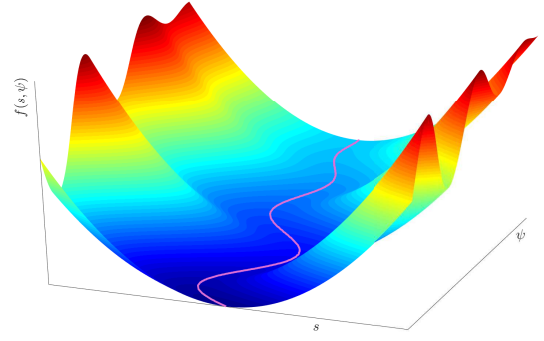}
	\caption{$f^\star(\psi)$ continuous at $\psi$}
	\label{fig:berge_theorem_illustration}
\end{figure}

\begin{remark}
	Theorem~\ref{thm:berge_illustration} states that $f^\star$ is continuous at each $\psi\in\Psi$. Fig.~\ref{fig:berge_theorem_illustration} illustrates the case when $s$ and $\psi$ are scalars. $f^\star$ is represented by the pink curve, which is like ``a winding stream running through high mountains''. It means that for each fixed $\psi$, $f^\star(\psi)$ is the minimum which can be found with respect to $s$. Moreover, the global minimum of $f(s,\psi)$ is on this pink curve. We only need to search along this continuous curve, and we will find the optimal attack strategy for this Byzantine attack optimization problem.   
\end{remark}

\subsection{Coordinate descent algorithm}
In the last subsection, we have proved that $f^\star$ is continuous at $\psi$, where $\displaystyle \psi=\left[\nu_0~~\nu_1~~\alpha\right]^\top$. With the Gaussian approximation method, the minimum of Problem~\ref{prob:original} then can be searched along $f^\star$ by numerical algorithms. Since we have only proved the existence of continuity for $f^\star$, other properties, i.e., differentiability and twice differentiability, are not guaranteed. Based only on the continuity, we propose Algorithm~\ref{alg:optimal-attack} to search the optimal Byzantine attack strategy for Problem~\ref{prob:correspondence}. The $\texttt{cvx}$ toolbox mentioned is a MATLAB-based modeling system for convex optimization. 
\begin{algorithm}[t]
	\small
	\caption{Coordinate Descent Algorithm for Optimal Byzantine Attack Strategy}
	\label{alg:optimal-attack}
	\begin{algorithmic}[1]
		\State Input: $T, \{a_k\}, \{b_k\}, \{c_k\}$
		\State Initialization: $\nu_0$, $\nu_1$, $\alpha\in(0,1)$;
		\State \texttt{cvx} Toolbox: compute $f^\star\left(\left[\nu_0~~\nu_1~~\alpha\right]^\top\right)$;
        \For{$k=1:1:T$}
        \State $\nu_{0}^-=\nu_0-a_k$;
        \State $\nu_{0}^+=\nu_0+a_k$;
        \State \texttt{cvx} Toolbox: compute $f^\star\left(\left[\nu_{0}^-~\nu_1~\alpha\right]^\top\right)$;
        \State \texttt{cvx} Toolbox: compute $f^\star\left(\left[\nu_{0}^+~\nu_1~\alpha\right]^\top\right)$;
        \If{$f^\star\left(\left[\nu_{0}^-~\nu_1~\alpha\right]^\top\right)\leq f^\star\left(\left[\nu_{0}^+~\nu_1~\alpha\right]^\top\right)$}
        \State $\nu_0\gets\nu_0^-$;~~$\text{flag}=-1$;
        \Else 
        \State $\nu_0\gets\nu_0^+$;~~$\text{flag}=1$;
        \EndIf
        \Repeat
        \State $\nu_0\gets\nu_0+\text{flag}\times a_k$;
        \State \texttt{cvx} Toolbox: compute $f^\star\left(\left[\nu_0~~\nu_1~~\alpha\right]^\top\right)$;
        \Until{$f^\star\left(\left[\nu_0~~\nu_1~~\alpha\right]^\top\right)$} does not descend;
        \State do Step $5$ ---- Step $17$ for $\nu_1$ and $\alpha$ with searching step lengths $b_k$ and $c_k$, respectively;
        \If{$f^\star\left(\left[\nu_0~~\nu_1~~\alpha\right]^\top\right)$ converges w.r.t. iteration $k$}
        \State break;
        \EndIf
        \EndFor		
	\end{algorithmic}
\end{algorithm} 

\section{Numerical Results}\label{sec:numerical}
In this section, we provide some numerical examples to illustrate the main results. We consider a scenario where the original probability measures $f_0$ and $f_1$ are distributed as:
\begin{align*}
\mathcal{H}_0:~~&f_0\sim \mathscr{N}\left(\mu_0=2,~\Sigma_0=2.8\right),\\
\mathcal{H}_1:~~&f_1\sim \mathscr{N}\left(\mu_1=10,~\Sigma_1=3.1\right).
\end{align*}
As shown in the first sub-figure in Fig.~\ref{fig:kl_divergence}, with the Gaussian approximation method, the KL divergence can be minimized by using the proposed coordinate descent algorithm when power constraint $\delta=80$. After $T=200$ iterations, a feasible attack solution is obtained as $\nu_0=11.9985$, $\nu_1=0.3385$, $\alpha=0.4069$, $\Gamma_0=2.8218$, $\Gamma_1=6.3137$, and a resulting KL divergence very close to $0$. This attack strategy is derived with the Gaussian approximation of the KL divergence objective. The real probability measures and the KL divergence between two Gaussian mixture models are portrayed in Fig.~\ref{fig:pdf}. It can be seen that the original KL divergence is $10.3251$ without Byzantine attack. By Monte Carlo sampling, which is introduced in Section~\ref{Monte_Carlo} with the sample size $K=100000$, the KL divergence under Byzantine attack is computed to be $0.8792$. The decrease of the KL divergence implies a tremendous increase of the missed detection error probability in the hypothesis testing as follows. Without the Byzantine attack, the false alarm probability $P_{\text{FA}}$ and the missed detection error probability $P_{\text{M}}$ under the Neyman-Pearson setup almost can be zero based on i.i.d. measurements from $10$ sensors. On the other hand, the designed Byzantine attack increases the missed detection error probability to $P_{\text{M}}^a=10.33\%$ while keeping the false alarm probability under $P_{\text{FA}}^a=0.04\%$.

The second sub-figure in Fig.~\ref{fig:kl_divergence} shows the approximated KL divergence curve with respect to the \textit{attacking power} $\alpha$ when constraint level $\delta=20$. For each fixed $\alpha$, we compute the KL divergence by using coordinate descent algorithm. We find that a larger \textit{attacking power} leads to a smaller KL divergence, which means a larger missed detection error probability. Notice that the KL divergence is still greater than $0$ even when $\alpha\geq 0.5$. This is because the Byzantine attack is launched by injecting noises instead of directly tampering measurements and it is conducted under an energy constraint. 

\begin{figure}[t]
	\centering
	\includegraphics[width=0.5\textwidth]{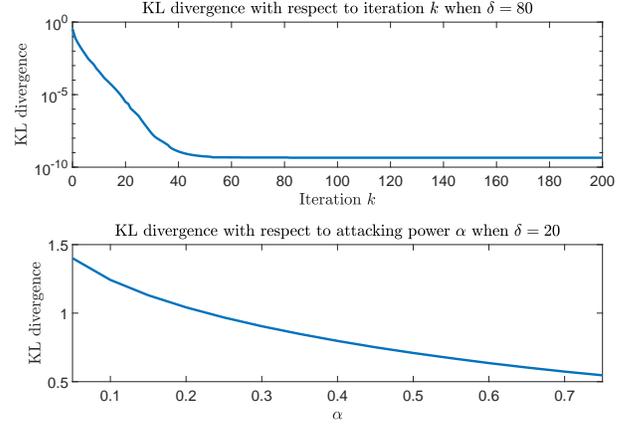}
	\caption{KL divergence w.r.t. iteration and \textit{attacking power}}
	\label{fig:kl_divergence}
\end{figure}

\begin{figure}[t]
	\centering
	\includegraphics[width=0.5\textwidth]{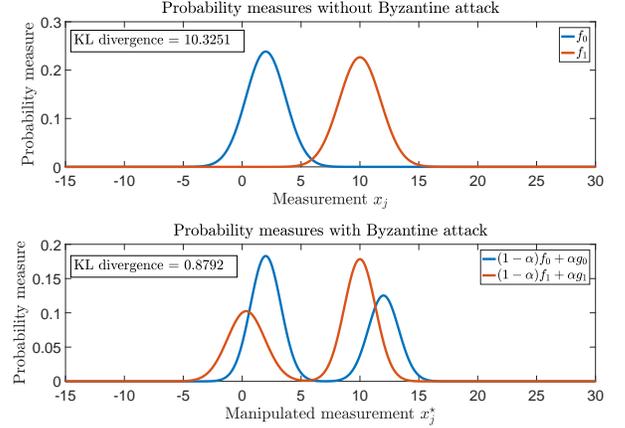}
	\caption{Probability measures without and with attack}
	\label{fig:pdf}
\end{figure}

\section{Conclusions}\label{sec:conclusion}
In this paper, a binary hypothesis testing is conducted based on measurements from a number of identical sensors, some of which may be compromised by a Byzantine attacker with probability $\alpha$. The attacker manipulates the measurements by injecting independent noises under the power constraint. We first formulated this attack optimization problem by using KL divergence to evaluate the attack impact. We then investigated the optimization problem with Gaussian approximation method and derived some theoretic results regarding the optimal attack strategy. In addition, a coordinate descent algorithm based on the theoretic results was proposed to search the optimal solution. Numerical examples verified the main results and showed the attack impact for the original problem, which is difficult to solve directly. Investigating this problem in vector case and with other approximation methods is a future direction.

\end{document}